\documentclass[12pt]{article}

\usepackage{amsthm}
\usepackage{amsmath}
\usepackage{amssymb}
\usepackage{enumerate}
\usepackage{xcolor}
\usepackage{listings}
\usepackage{multirow}
\usepackage{color}
\usepackage[pdfencoding=auto, psdextra]{hyperref}
\usepackage{doi}
\usepackage{booktabs}
\usepackage{rotating}
\usepackage{bm}
\usepackage{tablefootnote}
\usepackage[normalem]{ulem}
\usepackage{array}
\usepackage{longtable}
\usepackage{bbm}

\newcommand{\E}{\mathbb{E}}
\newcommand{\ed}{\mathrm{d}}
\newcommand{\R}{\mathbb{R}}
\newcommand{\id}{{\bf 1}}
\renewcommand{\P}{\mathbb{P}}
\newcommand{\Q}{\mathbb{Q}}

\newcommand{\SP}{\mathrm{SP}}
\newcommand{\AL}{\mathrm{AL}}
\newcommand{\TL}{\mathrm{TL}}

\newcommand{\CPI}{\mathrm{CPI}}

\newcommand{\DT}{\ed {\cal T}(t)}

\newcolumntype{L}[1]{>{\raggedright\let\newline\\\arraybackslash\hspace{0pt}}p{#1}}

\DeclareMathOperator{\spow}{sp}

\DeclareMathOperator{\argmax}

\newtheorem{theorem}{Theorem}
\newtheorem{lemma}[theorem]{Lemma}

\newtheorem{algorithm}[theorem]{Algorithm}
\newtheorem{inneraxiom}{Axiom}
\newenvironment{axiom}[1]
{\renewcommand\theinneraxiom{#1}\inneraxiom}
{\endinneraxiom}
\theoremstyle{definition}
\newtheorem{definition}[theorem]{Definition}
\numberwithin{equation}{section}
\numberwithin{theorem}{section}

\newcommand{\citep}[1]{\cite{#1}}

\begin{document}

\author{John Armstrong, Cristin Buescu}

\title{Collectivised Post-Retirement Investment}

\date{}              


\maketitle

\begin{abstract}
\noindent We quantify the benefit of collectivised investment funds, in which the assets of members who die are shared among the survivors. For our model, with realistic parameter choices, an annuity or individual fund requires approximately 20\% more initial capital to provide as good an outcome as a collectivised investment fund. We demonstrate the importance of the new concept of {\em pension adequacy} in defining investor preferences and determining optimal fund management. We show how to manage heterogeneous funds of investors with diverse needs. Our framework can be applied to existing pension products, such as Collective Defined Contribution schemes.

\end{abstract}

\noindent We study investment funds where all investors agree that any funds left when a member dies should be shared among the survivors. We call this new financial product a collectivised investment fund. An investor may choose to invest in such a collectivised fund at retirement instead of purchasing an annuity or managing an individual fund. We will argue that such funds should be particularly attractive to investors and will quantify the benefit of such a fund. 
While various types of risk-sharing funds have been proposed before, our proposal
differs by providing a clear framework to describe the optimal management of the fund, and by describing how risk should be shared.

While collectivised investment funds can be viewed either as an interesting financial product in their own right, they also provide a model for the post retirement phase of a Collective Defined Contribution (CDC) pension
scheme.
Such schemes are still very new and
the precise manner in which CDC funds are managed varies from scheme to scheme. Broadly
speaking a CDC fund is one which is managed for the benefit of a group of individuals
and endeavours to obtain a good pension for all its members, but which does not promise
a precisely defined pension. Examples of such schemes include the New Brunswick Hospitals' plan in Canada \citep{ppiCanada},
a number of Dutch pension schemes \citep{bovenberg,ppiDutch} and the planned new pension scheme for Royal Mail in the UK \citep{ppiCDC}.

CDC funds have emerged for two reasons. Firstly, low yields and tighter regulation have made traditional
employer pension schemes unattractive. Secondly, separate legislative changes, such as the Pension Schemes Act 2015 
in the UK,  allow much greater flexibility in how pensions
can be invested. 

Focusing on the UK as an example, pension funds have historically been either {\em defined-benefit} (DB)
funds or {\em defined-contribution} (DC) funds. In a DB fund an employer promises
to pay their employees and their partners a pre-specified income from retirement until death.
Often these benefits would be index linked, i.e.\ they would provide a constant real-terms income.
In a DC fund an employee has a personal fund of pension investments. Historically the assets
in a DC fund were used to purchase an annuity at retirement, i.e.\ a financial product that would
pay a pre-specified income to the retiree and their partner until death (typically a constant real-terms
income). As a result of the 2015 pension reform, it is now possible to receive the tax benefits
afforded to pension investment without being restricted to such a narrow range of investments.
The legislation introduces a new framework of ``defined ambition'' schemes into which 
CDC schemes fit.

In this paper, we will study the benefits of collectivised pensions to employees. We note, however, that CDC funds are
also very attractive to employers as they remove the liabilities inherent in a DB
scheme from their balance sheet, but we will not attempt to quantify this.  We note also that there will be legal and taxation considerations that one should take into account
in order to develop collective funds as a new financial product,
but these are jurisdiction specific and beyond the scope of this paper.

\medskip

There are a number of reasons why one would expect a collectivised investment fund to yield better pension outcomes than an annuity,
and for a CDC fund should outperform a DB fund.
\begin{enumerate}
	\item  \label{reason:equity} For large funds, ignoring systematic longevity risk, one can assume that
	an employer's DB liabilities depend only on interest rates. Assuming a typical risk-neutral
	pricing model for interest rate products, this means that a fully-funded DB scheme will not
	use equity investments, and so cannot benefit from the equity risk-premium.
	\item  \label{reason:intertemporalSubstitution}  A constant real-terms income does not benefit from the possibility of {\em intertemporal substitution}. This is the observation that if one is willing to delay consumption
	in favour of investing for longer, one may be able to obtain a higher rate of consumption in the future
	leading to a preferable outcome.
	\item  \label{reason:statePension} Due to 
	changes in the level of the state pension, the optimal deterministic real-terms income will change over time.
\end{enumerate}
We will also argue that there is an additional, less obvious, reason why constant consumption is
suboptimal.
\begin{enumerate}
	\setcounter{enumi}{3}
	\item  \label{reason:risk} A constant real-terms income ignores the risk of dying young and not enjoying any 
	consumption. It also ignores the risk of living on an inadequate pension for many years.
\end{enumerate}

A collectivised investment fund should also outperform an individual fund, and a CDC fund should also outperform
a DC fund. This is because 
collectivisation should reduce idiosyncratic longevity risk.

The primary research questions that this paper seeks to answer are: (i) How significant are these various effects are in practice? (ii) How should a fund be managed to best exploit these effects?

In order to highlight the effects of collectivisation, we will answer these research questions in the
context of the simplest market and longevity model that is capable of modelling all these effects.
Our market model will be a Black--Scholes--Merton model, and we will assume that future mortality distributions
are known so that there is no systematic longevity risk. For the same reasons, we consider only the post-retirement
investment and consumption and we do not model any form of bequests. Pension schemes are often intended
to provide for a couple in retirement. We model this as two individuals purchasing independent
shares in a collectivised fund.

\medskip

A first contribution of this paper is to compute
detailed numerical examples of pension outcomes based on reasonable market assumptions
in order to quantify the benefits of collectivisation. We find that either an annuity or
an individual fund requires $20\%$ more initial capital to be as effective as a collectivised fund.

A second contribution of this paper is an understanding of how to model investor
preferences over pension outcomes. The model we choose for investor preferences will play a critical
role in our theory. We will argue that the classical model for preferences, intertemporally additive
von Neumann--Morgernstern
utility, (used, for example, in \cite{merton1969lifetime}) are essentially risk-neutral. Conversely
annuities appear to be designed for the infinitely risk-averse. We find that using preferences
which allow for a moderate level of risk-aversion leads to much more plausible results, which
we believe will be more appealing to investors. We first study preferences from a theoretical
point of view  building on works such as \cite{vonNeumannMorgernstern},
\cite{krepsPorteus}, \cite{kihlstromMirman} and \cite{epsteinZin1} but
incorporating mortality.
We then compare the possible preference models numerically. Together our
theoretical and numerical results leads to a clear choice of the most
appropriate model for our problem. We will see that it is essential to use a preference
model which allows for flexible specification of what we call the {\em adequacy level}.
Our results also highlight the importance of risk-aversion,
leading to the somewhat unexpected phenomenon outlined in point \eqref{reason:risk}
above.

A third contribution of this paper is an algorithm for managing
heterogeneous funds, i.e.\ funds  containing individuals with diverse preferences, wealth
and mortality distributions. This is not a standard optimal control problem, as it is not clear how to 
define an objetive for the fund as a whole. Nevertheless we showed in \cite{ab-heterogeneous} that it
is possible to bound the utility that can be obtained by a collective fund. In the example
tested numerically in this paper, our proposed algorithm achieves $98\%$ of the maximum
potential benefit  of collectivisation for a heterogeneous fund of only $100$ members.

Together with the technical results of 
\cite{ab-heterogeneous}, these last two contributions provide
a rigorous framework for understanding collective pension investment.
The potential advantages of such pension products has been observed
before, leading to the development of with-profits annuities.
However, without a rigorous mathematical underpinning, it has been unclear how
to manage such products in the investor's best interests and as a result it has been unclear how
best to define contracts to guarantee that this is done \citep{smart}. We believe our framework is
capable of remedying these issues.

While we have studied only post-retirement investment, it follows
from our results that a CDC fund (managed optimally) will outperform a DB fund from the point of view
of the pensioner. To see this note that if one had a DB fund giving a guaranteed income, a collective
could sell this guaranteed income stream and then use the proceeds to pursue the optimal investment
strategy of this paper. We emphasize that this argument only applies to
an optimally managed CDC fund. Existing CDC funds are designed and managed using a
variety of heuristics rather than by solving an optimal investment problem.
Thus this paper has significant implications for the management of such funds.

\medskip

Having described our findings, let us now describe the structure of the paper.

In Section \ref{sec:prefs} we introduce the topic of preferences over pension outcomes,
before proposing two concrete models for preferences in Sections \ref{sec:exponential} and
\ref{sec:epsteinZinTheory}. In Section \ref{sec:model}
we select concrete market, preference and longevity models calibrated
to the UK pension market. In Section \ref{sec:comparison} we use this
to quantify the benefits of collectivisation for infinite funds
of identical investors using our preferred preference model. In Section \ref{sec:epsteinZin} we examine how
our results change if we use a different preference model, and use this
to identify the most appropriate choice to use in practice.
In Section \ref{sec:heterogeneousNumerics} we drop the assumption that
the fund contains infinitely many identical investors and give
an algorithm for managing heterogeneous funds. We test the efficacy
of this algorithm. We end in Section \ref{sec:conclusion}
with a summary of the financial consequences.

\medskip

\section{Preferences with Mortality}
\label{sec:prefs}

We must choose a model of investor preferences. The
choice of preferences will determine the meaning of optimal fund management, and so
this choice plays a central role in this paper. We should, therefore, justify our preference model in some detail.

If one ignores mortality for the moment, in the
current literature, homogeneous Epstein--Zin preferences seem to be the most popular model for preferences over consumption streams. These
preferences were introduced by \cite{epsteinZin1} and have been successfully applied to provide potential resolutions
to various asset pricing puzzles \citep{bansalYaron, bansal, benzoniEtAl, bhamraEtAl}. These preferences are {\em homogeneous} in 
the sense that if one prefers income $\gamma_1$ to $\gamma_2$ then one will prefer $\lambda \gamma_1$ to $\lambda \gamma_2$ for all positive
reals, $\lambda$. Homogeneity is a symmetry which ultimately results in a dimension reduction of the Hamilton Jacobi Bellman (HJB) equation yielding
a more tractable model.
Given the success of this
preference model, it is very natural to incorporate mortality into homogeneous Epstein--Zin preferences, and it is
easy to see how to do this in a manner that preserves the homogeneity property. This approach has already been taken in a number
of papers such as \cite{Gomes2005, Gomes2008,Bommier2017, Drouhin2015} which deal with mortality directly, as well as mathematical
works such as \cite{aurandHuang} which discusses random horizons in general.

However, we will argue that for pension investment problems, homogeneity of preferences is undesirable. We will define a
notion of {\em pension adequacy} for an individual's preferences. Any finite, non-zero choice of adequacy level
will automatically break the homogeneity of the preferences. We will
later confirm numerically that considering finite, non-zero adequacy levels
leads to quite different numerical results to those obtained with homogeneous preferences.

In Section \ref{sec:preferenceRelations} we will establish the necessary technical vocabulary for a theoretical comparison
by defining the notion of pension outcomes and of
preferences over such outcomes. We define a number of desiderata for a preference model
and define pension adequacy.

In Section \ref{sec:exponential} we identify a preference model from first principles which
we call {\em exponential Kihlstrom--Mirman preferences with mortality},
as it incorporates mortality into the preference models of \citep{kihlstromMirman,Kihlstrom2009}.
We will see that this model meets all our desiderata.

In Section \ref{sec:epsteinZinTheory} we present the alternative approach of homogeneous Epstein--Zin preferences
with mortality. We explain how this model could be modified to incorporate a flexible model for pension adequacy, but at
the cost of homogeneity.

\subsection{Preference relations over pension outcomes}
\label{sec:preferenceRelations}

We model a ``pension outcome'' as a pair $(\gamma,\tau)$ consisting of a stochastic process
$\gamma_t$, representing the rate at which payments are received at time $t$, and
a random variable $\tau$ representing the time of death. The underlying
filtered probability space will be denoted by $(\Omega, {\cal F}, {\cal F}_t, \P)$. The units
of $\gamma_t$ should be taken to be in real terms to ensure that our models
for inflation and preferences are separate.

We consider both discrete and continuous cashflow processes
$\gamma_t$. We write ${\cal T}$ for the set of time indices which
may be either $[0,T)$ or the evenly spaced time grid $\{ 0,\, \delta t,\,  2 \delta t,\,  3 \delta t, \ldots, T-\delta t \}$ where
$T$ is an upper bound on an individual's possible age which may be infinite. We write $\DT$ for the measure determined by the index set: this will be the Lebesgue measure on $[0,\infty)$ in the continuous case, or the sum of Dirac masses of mass $\delta t$ at each point in ${\cal T}$ for the discrete case. It will occasionally be convenient to allow the cashflow $\gamma_t$ to be non-zero when $t>\tau$, but this cash will not be consumed.
In the discrete case we assume that cashflow at the moment of death $\gamma_\tau$ is still consumed. So the total consumption
over the lifetime of an individual is
\[
\int_0^\tau \gamma_t \, \DT.
\]

We wish to describe an individual's preferences over pension
outcomes. This will be represented by an ordering $\preceq$ on the
set of pairs $(\gamma, \tau)$.
The outcome $(\gamma, \tau)$ is considered preferable to the
outcome $(\tilde{\gamma}, \tilde{\tau})$ if $(\tilde{\gamma}, \tilde{\tau}) \preceq (\gamma, \tau)$. We define $\succeq$ in the
obvious way and write $x \sim y$ if $x \preceq y$ and $y \preceq x$. We assume that an individual is indifferent to cashflows after death. This can be expressed mathematically as
\[
(\forall t\leq \tau \quad \gamma_t=\tilde{\gamma}_t )  \implies (\gamma_t,\tau) \sim
(\tilde{\gamma}_t,\tau).
\]

We will now define various properties that a preference relation may possess and which may be
considered desirable.

\begin{definition}
	The preferences are {\em monotonic} if $(\gamma, \tau) \preceq (\gamma^\prime, \tau)$ if $\gamma_t \leq \gamma^\prime_t$ for all $t \in {\cal T}$.	
\end{definition}
This simply reflects preference for consuming more.

We would like the preferences to depend only on the probabilistic properties of $\gamma$ and $\tau$ and not on
any extraneous data. We will formalize this requirement as the concept of {\em invariance}. To
define this, we first
recall that a mod 0 isomorphism is a measure preserving bijection from a full subset of a probability
space $\Omega$ onto a full subset of another probability space $\Omega^\prime$ with
measurable inverse. An automorphism
of a filtered probability space is a mod 0 isomorphism of probability
spaces that acts as a mod 0 isomorphism on each element of the filtration.

\begin{definition}
	The preferences $\preceq$ are \em{invariant} if for any 
	automorphism, $\phi$, of the filtered probability space $(\Omega, {\cal F}_t, \P)$ we have that $\phi$ preserves $\preceq$.
\end{definition}

Our framework is very similar to the descriptive model presented in Section 4 of \cite{krepsPorteus}, but differs
in that we consider preferences over random variables rather than preferences over distributions of random variables. Requiring
invariance acts as a substitute for defining the preferences over distributions. We will not repeat the axioms
of \cite{krepsPorteus}, but we note that the specific preference models we will ultimately use in this paper
will also satisfy their axioms.

\begin{definition}
	The preferences $\preceq$ are \em{law-invariant} if they depend only
	on the joint distribution of $(\gamma_t,\tau)$.	
\end{definition}

If preferences depend upon the time at which information becomes available, they may be invariant but
not law-invariant. \cite{krepsPorteus} developed their theory to allow the study of preference
relations which depend upon the timing of the resolution of uncertainty and \cite{epsteinZin1} provides
some discussion of when this may be desirable. However, for normative pension investment with no
exogenous investment opportunities, law-invariance seems a desirable feature as it is hard to justify why
one might be willing to pay to receive (or not receive) information which one cannot act upon.

\begin{definition}
	The preferences $\preceq$ are {\em convex} if for any $\gamma$ and $\tau$
	the set 
	\[
	\{ \tilde{\gamma} \mid (\gamma,\tau)\preceq(\tilde{\gamma},\tau) \}
	\]
	is convex.
\end{definition}
Convex preferences are mathematically desirable as one may then apply the tools of convex analysis. As
we shall see, convex preferences arise naturally as a consequence of the concepts of satiation and risk-aversion.

Given cashflows $\gamma_t$ defined on an interval
$t \in [a,b)$ and cashflows $\tilde{\gamma}_t$
defined on an interval $t \in [b,c)$ we define the concatenated
cashflow on $[a,c)$ by
\[
(\gamma \oplus \tilde{\gamma})_t=\id_{[a,b)}(t)\, \gamma_t + \id_{[b,c)}(t)\, \tilde{\gamma_t}.
\]

\begin{definition}
	The preferences $\preceq$ are {\em Markovian} if
	for any cashflows $\gamma_{\alpha,t}$, $\gamma_{\beta,t}$ defined on
	the finite interval $[0,a)$ with $a$ of measure zero (i.e.\ $a$ is not a grid point in the discrete case) and any cashflows $\gamma_{1,t}$,
	$\gamma_{2,t}$ defined on $[a,\infty)$
	\[
	(\gamma_\alpha \oplus \gamma_1, \tau)
	\preceq (\gamma_\alpha \oplus \gamma_2, \tau)
	\iff 
	(\gamma_\beta \oplus \gamma_1, \tau)
	\preceq (\gamma_\beta \oplus \gamma_2, \tau)
	\]
\end{definition}

This definition captures the case when future preferences do not
depend upon the past. There is no logical reason to insist that
preferences should behave in this way: for example if one has
purchased a house, the anticipated cost of housing repairs might
well affect one's future preferences.

Markovian preferences
are desirable mathematically because they result in more
tractable problems: if one has non-Markovian preferences then one must keep track of additional state variables
when solving optimal control problems and this increases the dimension
of the HJB equation. Markovian preferences are desirable from the point of view
of parsimony as one need not choose an initial state.

\begin{definition}
	The preferences $\preceq$ are {\em stationary} if for all $a$ there exists
	an isomorphism of filtered probability spaces
\[
\phi:(\Omega, {\cal F}, ({\cal F}_t)_{t\geq a}, \P) \to (\Omega, {\cal F}_a, \P) \times (\Omega, {\cal F}, ({\cal F}_t)_{t\geq 0}, \P)
\]
	such that for any cashflows $\gamma_{\alpha,t}$ defined on
	a finite interval $[0,a)$ with $a$ of measure $0$, any cashflows $\gamma_{1,t}$,
	$\gamma_{2,t}$ defined on $[0,\infty)$
	\[
	(\gamma_\alpha \oplus (\gamma_1\circ \phi), \tau \circ \phi + a)
	\preceq (\gamma_\alpha \oplus (\gamma_2 \circ \phi), \tau \circ \phi + a)
	\iff 
	(\gamma_1, \tau)
	\preceq (\gamma_2, \tau)
	\]
\end{definition}

This definition captures the case when preferences over future cashflows remain constant in time. The isomorphism $\phi$
is required in order to define preferences at future times in terms of preferences at time $0$.
Stationarity implies Markovianity.
Stationary preferences are particularly
parsimonious as one does not have to justify how the
preferences vary in time. Stationary preferences are very attractive in
infinite-horizon problems as they lead to a time-symmetry
of the HJB equation, which then allows the dimension to be reduced.

Our notion of stationary preferences corresponds to ``stationarity of preference'' in \cite{koopmans} (we say ``corresponds to'' because our set-up is slightly different). 
It is related to the concept called ``intertemporal consistency of preference'' (\cite{krepsPorteus,johnsenDonaldson}) and ``recursive preferences'' (\cite{epsteinZin1}). However, as we have not specified preferences at future times,
we have no need for an axiom of intertemporal consistency. The preferences at future times are described implicitly by preferences at time $0$ and the
requirement of temporal consistency (as explained by \cite{krepsPorteus}). Stationarity then requires that these implicit preferences at future times are isomorphic to the preferences
at time $0$.

\begin{definition}	
	An {\em adequacy level} for preferences ${\preceq}$ is
	a random process such that one is indifferent between dying at a particular time and living longer 
	while receiving an income at the adequacy 
	level. Formally, an ${\cal F}_t$-adapted, process $a_t$ is an {\em adequacy level} for the preferences $\preceq$ if
	\begin{enumerate}
		\item $\tau < \tilde{\tau}$;
		\item $\forall t\in[0,\tau] : \gamma_t=\tilde{\gamma_t}$;
		\item and $\forall t\in(\tau,\tilde{\tau}]: \tilde{\gamma}_t=a_t$.
	\end{enumerate}
	together imply $(\gamma,\tau) \sim (\tilde{\gamma},\tilde{\tau})$. If
	death is better than any finite cashflows we will say that the adequacy level is $\infty$. If death
	is worse than any finite cashflows we will say that the adequacy level is $-\infty$.
\end{definition}

\begin{definition}
	\label{def:vonNeumannMorgernstern}	
	{\em Inter-temporally additive von Neumann--Morgernstern preferences with mortality} are determined
	by a choice of concave, increasing utility function $u:\R \to \R$
	and a discount rate $b$. The preferences for $(u,b)$ on pension outcomes with non-negative
	cashflows are
	\begin{equation}
	(\gamma,\tau) \preceq (\tilde{\gamma},\tilde{\tau})
	\iff
	\E\left( \int_0^\tau e^{-bt} u( \gamma_t ) \DT \right) \leq \E\left( \int_0^{\tilde{\tau}} e^{-bt} u( \tilde{\gamma}_t ) \DT \right).
	\label{eq:vnmDef}
	\end{equation}
\end{definition}

This definition is based on \cite{vonNeumannMorgernstern}. These preferences are
montonic, convex, invariant, law-invariant, Markovian and stationary with an adequacy level of $u^{-1}(0)$.
In control problems where one cannot control mortality,
the adequacy level of these preferences is unimportant. To see why, observe that
\[
\E\left( \int_0^\tau (u(\gamma_t) + c ) \, \DT \right)
= \E\left( \int_0^\tau u(\gamma_t)\, \DT \right)  + \E\left( \int_0^\tau c \, \DT
\right).
\]
The term on the right is independent of the cashflows $\gamma$
and so the preferences are unchanged when one adds a constant $c$
to the utility function $u$. However, the adequacy level would become
important in problems where $\tau$ could be controlled, for example,
in a problem where one may choose to increase health-care expenditure
to increase life-expectancy.

Although inter-temporally additive von Neumann--Morgenstern preferences have many attractive properties,
we will
argue in the next section that they fail to adequately model risk-aversion. 
In models which include
risk-aversion, we will find that the adequacy level plays a role even if one cannot
control mortality.

\subsection{Exponential Kihlstrom--Mirman Preferences}
\label{sec:exponential}

We will now see how a number of simple considerations lead to a particular
form of preference model which we call exponential Kihlstrom--Mirman preferences.
We will suppose that there is an upper bound $T$ on the duration of an individual's life, so only
$\tau \leq T$ are considered admissible and we shall work with continuous time consumption.

Let us consider an individual's preferences over deterministic outcomes $(\gamma,\tau)$. We will assume
the individual is {\em order indifferent}, which we define to mean that their preferences depend only
on the distribution function of $\gamma$ over time defined by:
\[
F_{\gamma}(x) = \int_0^T \id_{\gamma_t \leq x} \, \DT.
\]	
Let us also suppose that their is a fixed, deterministic adequacy level $a$. Let us write
$\bm{a}$ for a constant income stream at the adequacy level. We may
out any deterministic consumption stream $(\gamma, \tau)$ to the right at the adequacy level to
obtain an equally preferable outcome $(\gamma \oplus \bm{a}, T)$. So an individual's preferences over
deterministic consumption streams are then determined by their preferences over 
the distribution functions of $\gamma \oplus \bm{a}$.

Preferences over distribution functions were studied by \cite{vonNeumannMorgernstern} who showed
that under modest axioms, preferences over distribution functions are determined by an expected utility.
Although they had probability distributions in mind, rather than temporal distributions,
the two problems are mathematically identical.

Rather than duplicate the axioms of von Neumann and Morgernstern, we propose a single axiom which captures
their results together with the notions of order indifference and a deterministic adequacy level.

\begin{axiom}
	Preferences between deterministic outcomes $(\gamma, \tau)$
	and $(\tilde{\gamma}, \tilde{\tau})$ are described by a utility function $u:\R \cup \{a \} \to \R$ with
	\[
	(\gamma, T)\preceq (\tilde{\gamma}, T)
	\iff
	s(\gamma_t, \tau)\leq s(\tilde{\gamma}_t, \tilde{\tau})
	\]
	where
	\[
	s_t:= \int_0^{\tau} u(\gamma_t) \, \DT.
	\]
	We call $s_t$ the {\em satisfaction} associated with $(\tilde{\gamma},\tilde{\tau})$.
	Note that $u$ is only determined up to scale.
	\label{axiom:orderIndifference}
\end{axiom}
\begin{axiom}
	An individual's preference over pension outcomes $(\gamma,\tau)$ are given by a von Neumann--Morgernstern
	preference relation over satisfaction.
	\label{axiom:vnmSatisfaction}
\end{axiom}

This discussion leads to the following definition.
\begin{definition}
	\label{def:kihlstromMirman}	
	{\em Kihlstrom--Mirman preferences with mortality} are determined by a
	choice of concave, increasing utility function $u:\R \to \R$, a second 
	increasing function $w:\R \to \R$ and a discount rate $b$.
	The preferences for $(u,w,b)$ on pension outcomes are
	\[
	(\gamma,\tau) \preceq (\tilde{\gamma},\tilde{\tau})
	\iff
	\E\left( w \left( \int_0^\tau e^{-bt} u( \gamma_t )\, \DT \right) \right)
	\leq \E\left( w \left( \int_0^{\tilde{\tau}} e^{-bt} u( \tilde{\gamma}_t )\, \DT \right) \right).
	\]
	{\em Von Neumann-Morgernstern preferences with mortality} arise in the special case $w(x)=x$.
	We will call the case $w(x)=-e^{-x}$ and $b=0$ {\em exponential Kihlstrom--Mirman preferences}.
\end{definition}
If one replaces $\tau$ with a deterministic time $T$ one obtains the preferences without mortality of \cite{kihlstromMirman}.
We see that if an individual's preferences satisfy Axioms \ref{axiom:orderIndifference} and \ref{axiom:vnmSatisfaction}
then their preferences must be Kihlstrom--Mirman preferences with mortality and  the discount rate $b$ must equal $0$.

\medskip

We believe that the most important assumption we have made is order indifference. This is
an important assumption for our normative pensions model. We make this assumption because we believe that
an individual's pension in old age should be given equal weight to their pension at retirement.

This assumption is controversial and so merits further discussion.
There are a number of reasons why one might
include discounting in a preference model. Firstly, one might use discounting as a proxy for directly modelling mortality.
This idea is justified in \cite{ab-exponential} where it is shown for a specific model that the force of mortality
and the discount factor play mathematically equivalent roles. However, our model includes mortality endogenously.
Secondly, in a descriptive model, 
one might use discounting to model an irrational bias towards early consumption. However, our model is normative.
Third,  one might use discounting to represent exogenous investment opportunities. However, we seek to model
the entire market endogenously.

As well as assuming that equal weight is given to all ages, the assumption of
order indifference requires that the utility function and the adequacy
level remain constant over time. We believe this is reasonable if 
one asks what
form a preferences should take in a model that
consciously chooses to ignore any features of pension outcomes other
than mortality and cashflows. However, it may be beneficial to relax these requirements
to allow for more flexible modelling. For example, we will allow the adequacy level to change over time in 
our numerical work in order to model a non-constant deterministic state pension.

\medskip

We will say that Kihlstrom--Mirman preferences are {\em monotone} if both $u$ and $w$ are monotone increasing,
so that preferences are increasing as a function of satisfaction, and satisfaction is increasing as a function of consumption. We will say that they model satiation if $u$ is concave, in which 
case there will be diminishing returns at higher levels of consumption as $u^{\prime}(c_1) \leq u^{\prime}(c_2)$ if $c_1 < c_2$.
The term ``satiation'' is non-standard
with most authors preferring to talk in terms of intertemporal substitutability (e.g.\ \cite{kihlstromMirman, epsteinZin1, duffieEpstein, xing}).
We prefer the term satiation partly because it is more intuitive and easier to say. It also refers
to the preferences themselves: by
contrast intertemporal substitution refers to the resulting behaviour when interest rates are changed and so incorporates
the market model into the terminology for preferences.

We will say the preferences are {\em satisfaction-risk-averse} if $w$ is concave.
This is the assumption that we would prefer to receive the satisfaction $\E(s)$ with certainty than to receive a random satisfaction $s$.
Since Axiom \ref{axiom:orderIndifference} presupposes that satisfaction, being an integral, has additive properties, it is reasonable to take expectations of
satisfaction. This is important because the concept of risk-aversion
is not topologically invariant and depends upon the additive
structure of $\R$.

There is an alternative additive structure one could consider, namely the structure defined by the additivity of cash values and this gives rise to an alternative concept of risk-aversion. Given
preferences satisfying our axioms, we may define the constant
cash equivalent of a deterministic cashflow $(\gamma,\tau)$ by
\[
c(\gamma,\tau)=u^{-1} \left( \frac{1}{T}
\int_0^\tau u(\gamma_t) \, \DT \right).
\]
We may then write our preferences over non-deterministic cashflows as
\[
(\gamma, \tau) \preceq (\gamma^\prime, \tau^\prime)
\iff
\E( w(T u(c(\gamma,\tau)) ) \leq \E( w(T u(c(\gamma^\prime,\tau^\prime)) ).
\]
This leads to the definition that these preferences are {\em monetary-risk-averse} if the function $x\to w(T u(x))$ is concave.

We see that Kihlstrom--Mirman preferences successfully separate an individual's satiation preferences
and an individual's risk preferences.

If one agrees that satisfaction-risk-aversion
is the correct operationalization of the intuitive concept of risk aversion,
one is lead to the conclusion that inter-temporally additive von Neumann--Morgernstern preferences do not model
risk aversion at all. This is a rather stronger statement than the more familiar observation that
inter-temporally additive von Neumann--Morgernstern preferences fail to disentangle risk aversion and
satiation \cite{duffieEpstein,xing}.

If one insists on Markovianity, then, as was observed in \cite{epsteinZin1}, one may identify the function $w$.
\begin{lemma}
	\label{lemma:exponentialKM}	
	Kihlstrom--Mirman preferences with mortality in continuous time are Markovian
	if and only if $w$ takes the either the form 
	\[ w(x)=c_1 \exp(c_2 x) + c_3 \quad \text{or} \quad w(x)=c_1 x + c_2\]
	for some constants $c_1,c_2,c_3 \in \R$ for $x \in U$ defined by
	on the set $U$ defined by
\[
U=(M \inf u, M \sup u), \quad M=\int_0^T e^{-bt} \, \DT
\]	
	They are stationary only if one additionally has $b=0$.
\end{lemma} 
\begin{proof}
	The function $w$ in Kihlstrom--Mirman preferences is determined by the
	preferences up to positive affine transformation. Hence	
	the preferences will be Markovian if and only if for any
	admissible $\gamma_{\alpha,t}$ and 
	$\gamma_{\beta,t}$ defined on $[0,a)$ and
	$\gamma_t$ defined on $[a, \infty)$ we can find $A>0$
	and $B$ such that
	\begin{multline}
	\E\left( w\left(\int_0^a e^{-bt} u(\gamma_{\alpha,t}) \DT + \int_a^\infty
	e^{-bt} u(\gamma_t) \, \DT \right) \right) \\
	= A
	\, \E\left( w\left(\int_0^a e^{-bt} \gamma_{\beta,t} \DT + \int_a^\infty
	e^{-bt} u(\gamma_t) \, \DT \right)\right) + B.
	\label{eq:affineRelation}
	\end{multline}
	The ``if'' statement for Markovian preferences is now clear.	
	
	To prove the ``only if'' part of the same statement, we may assume without loss of generality
	that $0 \in U$ since a shift in the definition of $u$ can be accommodated by
	the choice of constants. Let us choose deterministic $\gamma_{\alpha,t}$, $\gamma_{\beta,t}$
	and $\gamma_t$ and introduce variables $x$, $\epsilon$ and $y$
	\begin{align*}
	x+\epsilon &:= \int_0^a e^{-bt} u(\gamma_{\alpha,t}) \, \DT, \\
	x &:= \int_0^a e^{-bt} u(\gamma_{\beta,t}) \, \DT, \\
	y &:= \int_a^T e^{-bt} u(\gamma_{t}) \, \DT.	
	\end{align*}
	We may then rewrite equation \eqref{eq:affineRelation} as
	\[
	w( x + \epsilon + y ) = A w( x + y ) + B.
	\]
	The constants $A$ and $B$ may depend upon $x$ and $\epsilon$, but they are independent of $y$. 
	Taking $x$ as a fixed point in $U$,
	our assumption on $u$ allows us to choose $\epsilon$ to be arbitrarily small and to choose
	$y$ arbitrarily in some interval $I$ around $0$. Taking $y=(n-1)\epsilon$ we find
	\[
	w(x + n \epsilon ) = A w( x + (n-1)\epsilon) + B.
	\]
	Hence for $n$ such that $n \epsilon \in I$, we have
	\begin{align*}
	w(x + n \epsilon) = \begin{cases}
	A^n w(x) + \frac{1-A^n}{1-A} w(x) B & A \neq 1 \\
	w(x) + n B & A = 1.
	\end{cases}
	\end{align*}
	This gives the result on the grid of points in $I$ starting at $x$ separated by a distance $\epsilon$.
	The result for points of the form $x+q \in I$ for rational $\Q$ follows by refining the grid and the
	general case of $x \in U$ is now clear.
	
	We now specialise to the case where $w(x)=-\exp(-x)$.
	\[
	\E\left( w\left(\int_0^a e^{-bt} \gamma_{\alpha,t} \DT + \int_a^\infty
	e^{-bt} \gamma_t \, \DT \right) \right)
	= A
	\E\left( w \left( \int_0^\infty
	e^{-b(t-a)} \gamma_{t-a} \, \DT \right) \right).
	\]
	The preferences will be stationary if and only if this is equal
	to some affine transformation applied to
	\[
	\E\left( w \left( \int_0^\infty
	e^{-b t} \gamma_{t-a} \, \DT \right) \right).
	\]
	The result for stationary preferences follows.
\end{proof}	

Let us summarize the properties of exponential Kihlstrom--Mirman preferences with mortality.
\begin{lemma}	
Exponential Kihlstrom--Mirman preferences with mortality are the only continuous time preferences with mortality which
\begin{enumerate}
\item  satisfy Axioms \ref{axiom:orderIndifference} to \ref{axiom:vnmSatisfaction}
\item  are strictly risk-averse,
\item  and are Markovian.
\end{enumerate}
If $u$ is monotone increasing and concave, then exponential preferences are also monotonic, invariant, law-invariant, stationary and
convex.
\end{lemma}	

\subsection{Epstein--Zin preferences}
\label{sec:epsteinZinTheory}

Kihlstrom--Mirman preferences are not the most popular choice to model preferences in economic literature. Many
authors prefer to consider {\em Epstein--Zin preferences}, which we will describe in this section.

The theoretical observation that makes Epstein--Zin preferences more popular than Kihlstrom--Mirman preferences
is that
Lemma \ref{lemma:exponentialKM} shows Kihlstrom--Mirman preferences are not stationary if one incorporates
discounting. For example, in \cite{epsteinZin1} it is remarked:
\begin{quote}
	Finally, note that if indifference to timing [of information arrival] and the intertemporal consistency of preferences are both assumed, then (\cite{chewEpstein}) an expected utility ordering is implied. 	
\end{quote}
Although they do not emphasize discounting here, elsewhere in their paper, Epstein and Zin do make their
implicit assumption on $b$ clear. In our terminology, \cite{chewEpstein} show that discounting, law-invariance and stationarity are incompatible.

Our focus in this paper is on normative pension investment, but Epstein and Zin's focus is wider. For example,
they remark that when choosing a preference model,
ultimately one should ``let the data speak'' suggesting their priorities are descriptive.

Since discounting is crucial in many economic models, this motivated Epstein and Zin
to propose dropping law-invariance and so consider models where the time at which information is received is important.
\cite{epsteinZin1} describes a general theory of such stationary preferences,
extending the work of \cite{krepsPorteus} to the infinite time setting.
In this regard, their work extends the homogeneous preferences proposed by
\cite{chew} and \cite{dekel}. 
As we shall see, the resulting theory separates
satiation and risk in a very similar manner to Kihlstrom--Mirman preferences.

Although these considerations are important, a second reason for the popularity of Epstein--Zin preferences
is that they incorporate homogeneous preferences.  In this context, Lemma \ref{lemma:exponentialKM} tells us that Kihlstrom--Mirman
preferences cannot simultaneously have the symmetries of homogeneity (corresponding to using a power function for $w$)
and stationarity. This provides a motivation for considering homogeneous Epstein--Zin preferences even if
one believes that discounting is not required for the problems we are considering.

\medskip

The general form of Epstein--Zin preferences for a sequence of positive scalar cashflows $\gamma_t$ is given by
\begin{equation}
Z_t(\gamma) = [ \gamma_t^\rho + \beta \mu_t( Z_{t+\delta t}(\gamma) )^{\rho}]^{\frac{1}{\rho}}.
\label{eq:epsteinZinStandard}
\end{equation}
where $\mu$ is a certainty equivalent operator and $\rho \in (-\infty,1)\setminus \{0\}$ and $0<\beta<1$. Sometimes a normalization
constant $(1-\beta)$ is included in front of the $\gamma_t^\rho$, but this is not essential.

Since the sequence of cashflows $\gamma_t$ is infinite,
the equation \eqref{eq:epsteinZinStandard} only defines the utility
as the solution of a fixed point problem. The discount factor $\beta$
plays an important role in the proof that the fixed point exists.

Given an adequacy level $a$ and a pension outcome $(\gamma,\tau)$
we define $\gamma^a_t$ to the stream of cashflows equal to $\gamma_t$ up to death and $a$ after death. We may then define the Epstein--Zin
utility with mortality to be given by the standard Epstein--Zin
utility of $\gamma^a$.

We will be primarily interested in the case where 
\[
\mu_t( Z_{t+1}(\gamma) )=\E_t( Z_{t+1}(\gamma)^\alpha )^{\frac{1}{\alpha}}
\]
where $\alpha \in (-\infty,1)\setminus \{0\}$. We refer to this case as homogeneous Epstein--Zin preferences as they have the property that for $\lambda>0$
\[
Z_t(\lambda \gamma) = \lambda Z_t(\gamma)
\]
This symmetry yields a dimension reduction of the HJB, as shown in some generality in  \cite{xing,aurandHuang}.
This allows some interesting pension problems to be solved analytically, as demonstrated in \cite{campbellViceira}. In \cite{ab-ez},
a companion paper to this article, we use
symmetry to compute the optimal investment
strategy for homogeneous Epstein--Zin preferences in the Black--Scholes model for both individual and collective investment funds when consumption occurs in discrete time. 

If we assume that all individuals will eventually die, we may give a simpler definition for homogeneous Epstein--Zin preferences with mortality which has the additional advantage of allowing the case $\beta=1$.
\begin{definition}
	\label{def:homogeneousEpsteinZin}	
	{\em Homogeneous Epstein--Zin utility with mortality} is defined
	in discrete time and depends on parameters $\alpha \in (-\infty,1) \setminus\{0\}$, $\rho \in (-\infty,1) \setminus\{0\}$,
	and $0<\beta = e^{-bt} \leq 1$.
	It is the $\R_{\geq 0} \cup \{ \infty \}$-valued stochastic process
	defined recursively by
	\begin{equation}
	Z_t(\gamma, \tau) =
	\begin{cases}
	0 & t > \tau; \alpha>0 \\
	\infty & t > \tau; \alpha<0 \\
	\left[ \gamma_t^\rho + \beta \, \E_t( Z_{t+\delta t}(\gamma, \tau)^\alpha )^\frac{\rho}{\alpha} \right]^\frac{1}{\rho} & \text{otherwise}.
	\end{cases}
	\label{eq:defepsteinzin}
	\end{equation}
	To interpret this formula we use the convention $\infty^\alpha=0$ for $\alpha<0$.
	Assuming $\gamma_0$ is deterministic, $Z_0$ is deterministic,
	so we may define homogeneous Epstein--Zin preferences with mortality
	by
	\[
	(\gamma, \tau) \preceq (\tilde{\gamma},\tilde{\tau})
	\iff
	Z_0(\gamma, \tau) \preceq Z_0(\tilde{\gamma}, \tilde{\tau}).
	\]
\end{definition}
Our definition has been chosen so that the defining equation \eqref{eq:epsteinZinStandard} also holds. Our choice of value for the utility when $t > \tau$ is determined by the requirement that
positive homogeneity still holds.

Although the defining formula \eqref{eq:epsteinZinStandard} is
elegant, Epstein--Zin preferences are a little easier to understand if one
defines the signed power function by
\[
\spow_\gamma(x) = \begin{cases}
x^\gamma & \text{ when } \gamma>0 \\
-x^\gamma & \text{ when } \gamma<0
\end{cases}
\]
and defines the {\em Epstein--Zin satisfaction}, $z_t$ by
\[
z_t = \spow_{\rho}( Z_t, \rho ).
\]
We may then write the defining equations of homogeneous Epstein--Zin preferences as follows
\begin{equation}
z_t = \gamma_t^\rho + \beta \spow_{\frac{\alpha}{\rho}}^{-1} \left( \E_t\left( \spow_{\frac{\alpha}{\rho}}( z_{t+\delta t} ) \right) \right).
\label{eq:epsteinZinRewritten}
\end{equation}
Written in this form it becomes clear that the Epstein--Zin
satisfaction is an additive quantity (as indicated by the plus sign).
For deterministic cashflows, these preferences simplify to
\begin{equation*}
z_t = \gamma_t^\rho + \beta z_{t+\delta t}
= \sum_{i=0}^\infty \beta^i \gamma_{i \delta t}^\rho. 
\end{equation*}
Hence $\rho$ is a parameter measuring satiation.
We also see from \eqref{eq:epsteinZinRewritten} that the combination of parameters $\frac{\alpha}{\rho}$ can be interpreted as  satisfaction-risk-aversion parameter. The preferences are
satisfaction-risk-averse if $\frac{\alpha}{\rho}\leq 1$, i.e.\ if $\alpha<\rho$. In the case that $\alpha=\rho$ the preferences
are satisfaction-risk-neutral and degenerate to inter-temporally additive von Neumann--Morgernstern preferences.

With this interpretation in place, we may return to the Epstein--Zin utility itself $Z_t$. We now see that this is the instantaneous cash equivalent value of $z_t$. Hence $Z_t$ can be interpreted as a cash
value and we see that the parameter $\alpha$ is a monetary-risk-aversion parameter.

We note that the choice of utility value for $\tau > t$ is forced
upon us by the requirement that our preferences are positive
homogeneous and independent of any cashflows that occur after death.
This is a limitation of homogeneous Epstein--Zin preferences with
mortality. For $\alpha<0$ we must always assume that being dead
is preferable to any cashflow, for $\alpha>0$ we must assume
that any cashflow is preferable to being dead. These are
both extreme positions to take on pension adequacy. Moreover,
it is unfortunate that this view on pension adequacy cannot
be taken independently from one's monetary-risk aversion.

Note that in the situation where $\alpha=\rho$ the pension adequacy level will not affect investment decisions.

We summarize the properties of these preferences.
\begin{lemma}
	Homogeneous Epstein--Zin preferences with mortality are monotone, convex, invariant,
	Markovian and stationary,
	but are only law-invariant if $\alpha=\rho$.
\end{lemma}
The key advantages of these preferences are analytic tractability and the potential
to include discounting.

We have only described Epstein--Zin preferences in discrete time, but one may also formulate
a continuous time theory \cite{duffieEpstein, schroderSkiadas, xing, aurandHuang}. However, this theory is
considerably more complex than the continuous time theory for exponential Kihlstrom--Mirman preferences.

\medskip

It is
instructive to note that discrete time exponential Kihlstrom--Mirman preferences satisfy a similar equation to \eqref{eq:epsteinZinRewritten}.
If we define
\[
\tilde{z}_t:=-\log\left(\E\left(\exp\left(\int_0^\tau u(\gamma_t) \DT \right)\right)\right).
\]
then one easily checks that
\begin{equation}
\tilde{z}_t = u(\gamma_t \, \delta t) + v^{-1} \left( \E_t\left( v( \tilde{z}_{t+\delta t} ) \right) \right)
\label{eq:genRecursive}
\end{equation}
where $v(x)=-\exp(-x)$. Thus these preferences also fit into the recursive preferences
framework of \cite{krepsPorteus}, but are inhomogeneous.

To define Epstein--Zin type preferences with mortality that allow a flexible specification
of the adequacy level there are two approaches. Either one could define the utility via equation
\eqref{eq:genRecursive}, taking
$u(x)=x^\rho + \alpha$ for some $\alpha$ and $v=\SP_\frac{\alpha}{\rho}$. Alternatively
one could leave the refining recursion equation unchanged and modify the value of $a$, the
utility when dead. This latter
approach would require choosing a value of $\beta$ less than $1$. Since it is difficult to decide how to do this
for our normative investment questions, the former approach seems preferable.
Whichever approach one takes, breaking homogeneity is inevitable.

\bigskip

This completes our theoretical discussion of preference models. To make a final decision on which 
model is the most appropriate we  must wait until we can examine
our numerical results. Then we will know which models yield reasonable investment/consumption
strategies.

\section{A Realistically Parameterised Model}
\label{sec:model}

We wish to compare numerically the performance of annuities, individual funds and collective funds.
To do this we must now choose precise market, mortality and preference models. We choose
market and mortality models which are as close as possible to the model used in \cite{thefuturebook},
which in turn is based on modelling assumptions of \cite{obr2019}.
We use a version of exponential Kihlstrom--Mirman preferences modified to incorporate
a deterministic state pension.

We will work in continuous time for investments, but consumption will be assumed to take
place in discrete time, with $\delta t$ taken to be 1 year.

We specialise to the case of the Black--Scholes--Merton model. That is,
we suppose that there is a risk free asset $S^1_t$ growing at a risk free rate $r$ and a  risky asset $S^2_t$ which follows geometric Brownian motion
with drift $\mu$ and volatility $\sigma$:
\begin{align}
\ed S^1_t &= S^1_t( r \, \ed t ), &S^1_0 \nonumber \\ 
\ed S^2_t &= S^2_t( \mu \, \ed t + \sigma \, \ed W_t), & \quad S^2_0.
\label{eq:blackScholesMerton}
\end{align}
We emphasize that all values are quoted in real terms. In particular $r$ is the difference between the nominal interest rate and
the rate of inflation. Similarly $\mu$ measures real returns.

There are many well-known limitations to the Black--Scholes--Merton model. We believe that the most important limitation of the model for pension modelling is the assumption of a fixed deterministic interest rate. For example, the low interest rates that have prevailed over the last decade have had a dramatic impact upon pension outcomes. Nevertheless for the purposes of this paper (estimating the potential benefits of collectivised pensions) we believe that this limitation is acceptable. We aim to extend our approach to stochastic interest rate models in future research.

\medskip

Currently in the UK, the state pension grows in real terms due to the so-called ``triple lock''.
The UK Office of Budget Responsibility uses a deterministic model of the state pension growing at a
rate of average earnings growth plus $0.36\%$, yielding a net growth rate of $4.7\%$.

We are able to incorporate this into our preference model by choosing a gain function of the form
\begin{equation}
{\cal J}(\gamma,\tau)=\E\left(-\exp\left(-s_{\gamma,\tau} \right) \right)
\end{equation}
where the satisfaction, $s_{\gamma,\tau}$ is given by 
\[
s_{\gamma,\tau}=\sum_{t \in {\cal T}, t \leq \tau} u(\gamma_t,t)\, \delta t
\]
and where $u$ is given by
\begin{equation}
u(\gamma,t)=\begin{cases}
a(\gamma_t+\SP_t)^\rho - a (\AL_t+\SP_t)^\rho & \gamma_t\geq 0 \quad \forall t, \\
-\infty & \text{otherwise}.
\end{cases}
\label{eq:uParameterised}
\end{equation}
The parameter $\SP_t$ is a deterministic state pension at time $t$,
and $\AL_t$ is the adequacy level for the private pension. Thus if the individual consumes at a rate $\gamma_t=\AL_t$
at all times, their overall satisfaction will be $0$.
The parameter $\rho<1$ is a satiation parameter and $a$ is a satisfaction-risk-aversion parameter. If $\rho<0$, $a<0$ otherwise $a>0$.
Our gain function ensures that consumption must always be non-negative.

We note that incorporating the state pension into the model will inevitably break any homogeneity properties of the problem and, if the state pension is time varying, this will also break any translation
invariance properties. This is why we have selected to use 
a time varying-version of exponential Kihlstrom--Mirman preferences in this model.

\medskip

The numerical value of the parameter $a$ in our gain function will depend upon the units of currency. To remedy this
we first define $X_{\AL}$ to be the currency value required at time $0$ in order to fund a deterministic pension of $\overline{AL}_t:=\max\{\AL_t,0\}$
\[
X_{\AL} = \int_0^\tau e^{-rt} \overline{AL}_t \, \DT.
\]
We now define a parameter $\lambda$ by
\[
\lambda = \lim_{\epsilon\to 0} \frac{ s_{(1+\epsilon)\overline{AL}_t,t} - s_{\overline{AL}_t,t}}{\epsilon}.
\]
The parameter $\lambda$ therefore measures the rate of increase in satisfaction as one proportionately increases a deterministic pension set at the adequacy level. It provides a dimensionless parameter proportional to $a$.

We must choose all the parameters in our model in order to perform the comparison. 

Table \ref{table:parameterSummary} contains a summary of all our parameter assumptions.

\begin{sidewaystable}[htbp!]
	\centering	
	\begin{tabular}{L{0.1\linewidth}L{0.25\linewidth}L{0.25\linewidth}L{0.25\linewidth}} \toprule
		Parameter & Value & Description & Justification \\ \midrule
		$\SP_t$ & $\pounds 6718\, \exp(r_{\TL}t)$ & Annual state pension. & $\pounds 129.20$ per week as of July 2019. Growth as in OBR 2019. \\
		$\AL_t$ & $\pounds 16800-\SP_t$. & Adequacy level. & \cite{pensionCommission2004}
		and \citep{hmrc2019} \\
		$r_{\CPI}$ & $0.02$ & CPI growth. & OBR 2019. \\
		$r$ & $0.047-r_{\CPI}$ & Gilt returns. & OBR 2019. \\
		$r_{\TL}$ & $0.047-r_{\CPI}$ & State pension growth. & OBR 2019. \\
		$\mu$ & $0.082-r_{\CPI}$ & Equity growth. & Based on $4.6\%$
		index growth as in OBR 2019 together with a dividend yield of $3.6\%$
		as assumed in \cite{thefuturebook}.\\
		$\sigma$ & $0.15$ & & \cite{bloomberg} \\
		$X_0$ & $X_\AL=\pounds 126636$ & Initial fund value. & \\
		$\rho$ & $-1$ & Satiation parameter. & \cite{havranek2015}. \\ 
		$\lambda$ & $1$ & Satisfaction-risk-aversion parameter. & Illustrative choice.  \\ 
		$p_t$ & CMI\_2018\_F [1.5\%] & Mortality distribution & CMI 2018 \\ \bottomrule			
	\end{tabular}
	\caption{Parameter summary}
	\label{table:parameterSummary}
\end{sidewaystable}

The market parameters are mostly calibrated using the assumptions of \cite{obr2019}.
For equity returns we used the assumptions of the report
\cite{thefuturebook} which were designed to be compatible with those of the OBR.
To estimate equity volatility, $\sigma$, we used data for the FTSE
All Share Total Return Index from December 1985 to June 2019 obtained from \citep{bloomberg}.

The mortality distribution $p_t$ was obtained using the model CMI 2018 described in \citep{cmi2018}. We used this model to find the mortality distribution for women of UK retirement aged 65 in 2019 (65 being the UK retirement
age as of 2019). The CMI model requires one to choose a parameter determining the long-term rate of mortality improvement. We chose a long-term rate of $1.5\%$, the same value used in the illustrative examples of \cite{cmi2018}. To avoid numerical problems caused by low probability events, 
the age distribution was cut off at the point where the probability of surviving to this age was only $10^{-5}$. The resulting distribution is shown in Figure
\ref{fig:cmi2018f}.

\begin{figure}[htb]
	\centering
	\includegraphics{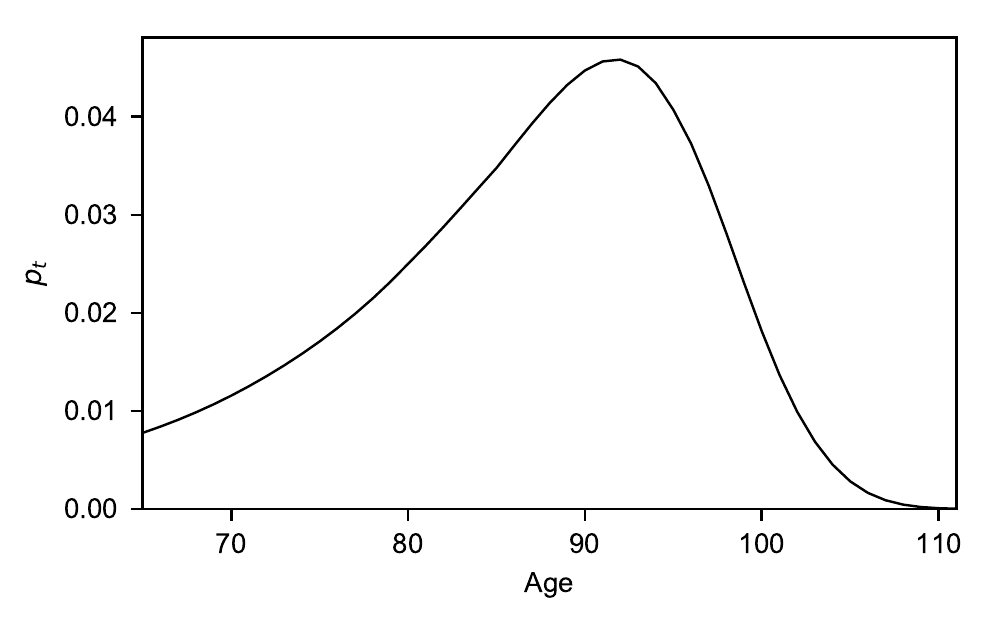}
	\caption{Probability density $p_t$ for the random variable $\tau$. Data for UK women aged 65 in 2019 using the model CMI\_2018\_F [1.5\%] (see \cite{cmi2018}).}
		\label{fig:cmi2018f}
	\end{figure}

The parameters determining the utility function are subjective and will vary from individual to individual, as will the available budget. Thus we can only choose illustrative values for these
parameters. We will now briefly explain how the values for these parameters were selected.

\begin{enumerate}
\item Choice of $\rho$. We know that in the case of von-Neumann Morgernstern preferences with utility function $u(x)$, the value of $\rho$ is closely related to the elasticity of inter-temporal substitution. Since von-Neumann Morgernstern preferences are a limiting case of exponential Kihlstrom--Mirman preferences, this suggests
we calibrate $\rho$ from empirically observed inter-temporal  substitution.
The mean elasticity observed in the meta-analysis \cite{havranek2015} is $0.5$. 
We compute the elasticity of inter-temporal substitution in the case of homogeneous Epstein--Zin
preferences for our market model in \cite{ab-ez}. Together these results  suggests we choose $\rho$ such that $\frac{1}{1-\rho}=0.5$, hence we take
$\rho = -1$. We emphasize that the value of $\rho$ will likely vary between individuals. The standard methodology
used to estimate the elasticity of inter-temporal substitution assumes a simple market model and is compatible with our
choice of the Black--Scholes model. For more sophisticated models with time varying market price of risk \cite{bansalYaron}
suggest the the elasticity of inter-temporal substitution may be closed to $1.5$ leading to the choice $\rho=\frac{1}{3}$. By choosing the smaller value $\rho=-1$ for our calculations we are erring on the
side of underestimating the benefits of a collectivised scheme over an annuity.
\item Choice of $X_0$. We choose the initial budget $X_0$ to equal $X_{\AL}$. Thus our illustrative individual can just afford a deterministic pension at the adequacy level.
\item Choice of $\AL_t$. We are referring to the parameter $\AL_t$ as the adequacy level because 
it is the obvious generalization of the notion of adequacy level given in Section \ref{sec:prefs}
to the form of gain function we are using in this section. However, the
term ``adequacy'' has already been used in the pension literature and we will insert quotation marks
around the word ``adequacy'' when the term should be understood in this broad sense. 

Various definitions for ``adequacy'' have been proposed. For example, one may choose an ``adequacy'' level based
on absolute poverty worldwide, relative poverty within one's country or relative to one's own lifetime earnings. See \cite{redwood2013} for a fuller discussion. There is no a priori reason why our formal notion of adequacy should correspond
to any particular notion of ``adequacy''. Indeed most notions of ``adequacy'' depend only on the age, nationality and income of
the individual whereas our notion of adequacy depends on preferences and so is likely to vary between individuals of identical age,
income and nationality. Nevertheless, we will choose one specific model of ``adequacy'' to determine $\AL$: specifically we will use the target replacement rates given in \cite{pensionCommission2004} to
determine the ``adequacy'' level as a proportion of final earnings.

The usual notion of ``adequacy'' refers to the required total pension. Since we have modelled the state pension by making a horizontal shift of our utility function, our notion of adequacy is correspondingly reduced by the state pension.

With this understood, we assume that our individual is earning $\pounds 24,100$ per annum, which is the median income before tax for individuals age 60--64 in 2016--17 in the UK
\citep{hmrc2019}. Then following the Pension Commission's suggested target replacement rates (\cite{pensionCommission2004}, updated to 2017 terms) we choose a target replacement rate of $70\%$. This gives an ``adequacy'' level of $\pounds 16,800$ per annum for the total
income from private and state pension.
\item Choice of $\lambda$. We take $\lambda=1$ as an illustrative example. To decide on a reasonable value for $\lambda$, one can look at the resulting range in the level of consumption when one simulates the investment strategy. We will plot a fan-diagram of the consumption in the next section (Figure \ref{fig:obrAssumptionsFan}) and it
can be seen from this diagram that $\lambda=1$ gives a reasonable result. In practice one might try to calibrate $\lambda$ for an individual using a risk questionnaire, but we will not attempt to consider
how such a questionnaire could be designed. 
\end{enumerate}

\subsection{Numerical comparison of annuities, individual investment and collectivised funds}
\label{sec:comparison}

Using the model with parameters as described in Section \ref{sec:model}
we are able to compute the optimal consumption for an individual fund $(n=1)$, a collective fund
$(n=\infty)$ and to compare this with an annuity. The problem may be written formally as
an optimal control problem, and this is done in \cite{ab-exponential} and, moreover,
that paper describes a numerical method to solve the problem. The resulting pattern of consumption is
shown in Figure \ref{fig:obrAssumptionsFan}.

\begin{figure}[htb]
	\centering
	\includegraphics{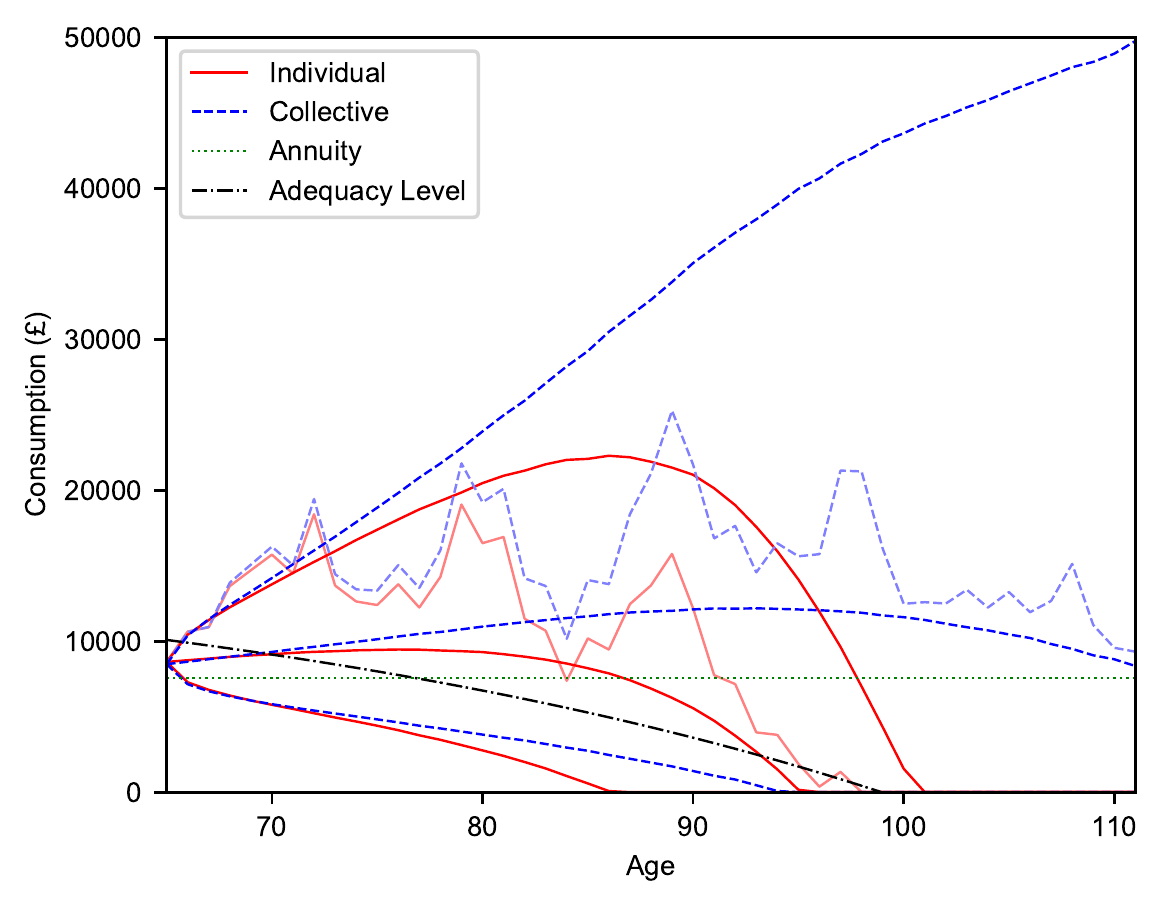}
	\caption{Fan diagram of optimal consumption over time for different types of pension fund.
	The percentiles shown in the fan are at $(5\%,50\%,95\%)$. }
	\label{fig:obrAssumptionsFan}
\end{figure}

The line illustrating the consumption of an annuity is straightforward to understand. It is a horizontal line as the consumption from an annuity is constant until death. Similarly the line representing pension adequacy is simple to understand. It starts at the current adequacy level but decreases over time deterministically due to the assumed increase in the state pension.

The optimal consumption of the individual fund, however, is not deterministic. To illustrate this consumption we have plotted a {\em fan diagram}. This fan diagram is made up of the four different lines in 
the figure all drawn with a continuous line style (and coloured red in colour reproductions).
Three of these lines represent the $5$th, $50$th and $95$th percentiles of the consumption at each point in time: these three lines are nearly smooth. The jagged line represents one illustrative random scenario. 

The optimal consumption for the collective fund is given by a similar fan diagram. The same stock price path was used to generate the illustrative random scenario for the collective and individual cases.

Since the gain of an individual depends on both consumption and mortality, one should cross-reference 
the diagram of consumption with Figure \ref{fig:cmi2018f} which shows the corresponding mortality distribution.

The diagram was obtained by computing a numerical approximation to the optimal investment
strategies using the method \cite{ab-exponential}. The percentiles were then estimated
by performing $10^5$ independent stock price simulations, applying the strategy and then computing the sample percentiles.

In Table \ref{table:annuityOutperformance} we present the {\em annuity equivalent} value of each investment-consumption approach. We define this to be the price of an annuity which would
give the same gain. This is a monetary measure of how much a strategy outperforms an annuity (or underperforms). We also present the {\em annuity outperformance}. This is the defined by
\[
\text{annuity outperformance} := \frac{\text{annuity equivalent}}{\text{budget}}-1.
\]
This gives a measure of the performance of the strategy relative to an annuity of
the same cost.

\begin{table}[h!tbp]
	\begin{center}
		\begin{tabular}{lrr}
			\toprule
			Fund & Annuity Equivalent ($\pounds\times 10^3$)& Annuity Outperformance \\ \midrule
			Annuity & 126.6 & 0\% \\
			Individual & 128.7 & + 1.5\% \\
			Collective & 152.2 & + 20\% \\
			\bottomrule 
		\end{tabular}	
		\caption{Numerical comparison of the three investment-consumption strategies
		with parameters as given in Section \ref{sec:model}}	
		\label{table:annuityOutperformance}
	\end{center}
\end{table}	

Our conclusion is that, for this illustrative example, collectivised pension investment
substantially outperforms an annuity. Although in this particular example, an individual
fund outperforms an annuity, a change to the parameters (for example taking $\mu=0.75$)
may yield a situation where the annuity outperforms the individual fund. By contrast
the optimum collective pension investment is guaranteed to outperform an annuity.

\subsection{The impact of satiation-risk-aversion and the adequacy level}

To understand how satiation-risk-aversion affects the optimal investment strategy in the
collectivised case, it is instructive to set the parameter $\lambda$ to a high value and to
increase the initial budget.
Figure \ref{fig:highLambda} illustrates the optimal consumption pattern if we set $\lambda=50$
and $X_0=2 X_{\AL}$, but otherwise use the same parameters as described in Section \ref{sec:model}.

\begin{figure}[htb]
	\centering
	\includegraphics{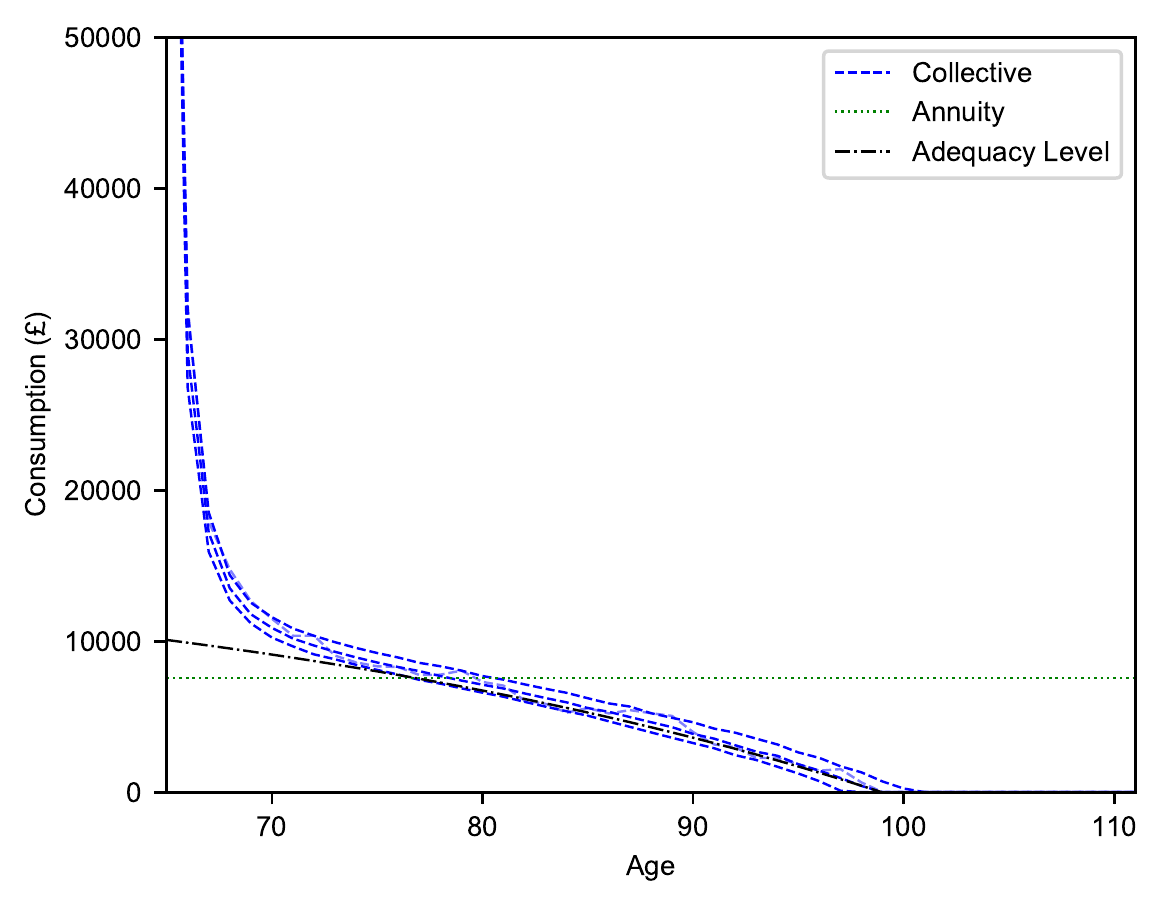}
	\caption{Fan diagram of the optimal consumption for an infinite collective for high satisfaction-risk-aversion $\lambda$. The percentiles shown in the fan are at $(5\%,50\%,95\%)$.}
	\label{fig:highLambda}
\end{figure}

We see that the initial consumption is very high, but then reduces to closely track
the adequacy level. To interpret this result, note that the only way it is possible to achieve
a deterministic satisfaction is to: start with a budget of at least the $\AL$; consume at the adequacy level at all times $t>0$; consume at the adequacy level plus any excess budget at time $t=0$. 
Given that this is the only strategy that yields deterministic satisfaction,
it is now unsurprising
that if the satisfaction-risk-aversion is set to a high value, the resulting consumption strategy will closely approximate this deterministic strategy.

If the initial budget is lower than $X_{\AL}$ but the risk-aversion is still high, we found in numerical examples that the behaviour was to consume at a low level until there is sufficient budget to begin tracking the adequacy level.

This behaviour is consistent with that found analytically for homogeneous
Epstein--Zin utility in \cite{ab-ez}. The behaviour in this case
is exaggerated because the adequacy level is forced to be either $0$ or $\infty$ in order to
achieve homogeneous preferences.

Our conclusion is that the adequacy level does indeed play an important role in
pension investment. We
note that a satisfaction-risk-averse individual may well decide to spend a large part of their
pension fund shortly after retirement due to the unhedgeable risk that they may die young.
We note that investing in an annuity suggests an inconsistent attitude towards risk: one is
being entirely risk-averse in investment decisions, yet one is ignoring the risk of dying young.
This may help explain why many pension investors instinctively find annuities unattractive.

\subsection{Comparison with Epstein--Zin preferences}
\label{sec:epsteinZin}

\begin{figure}[htb]
	\centering
	\includegraphics{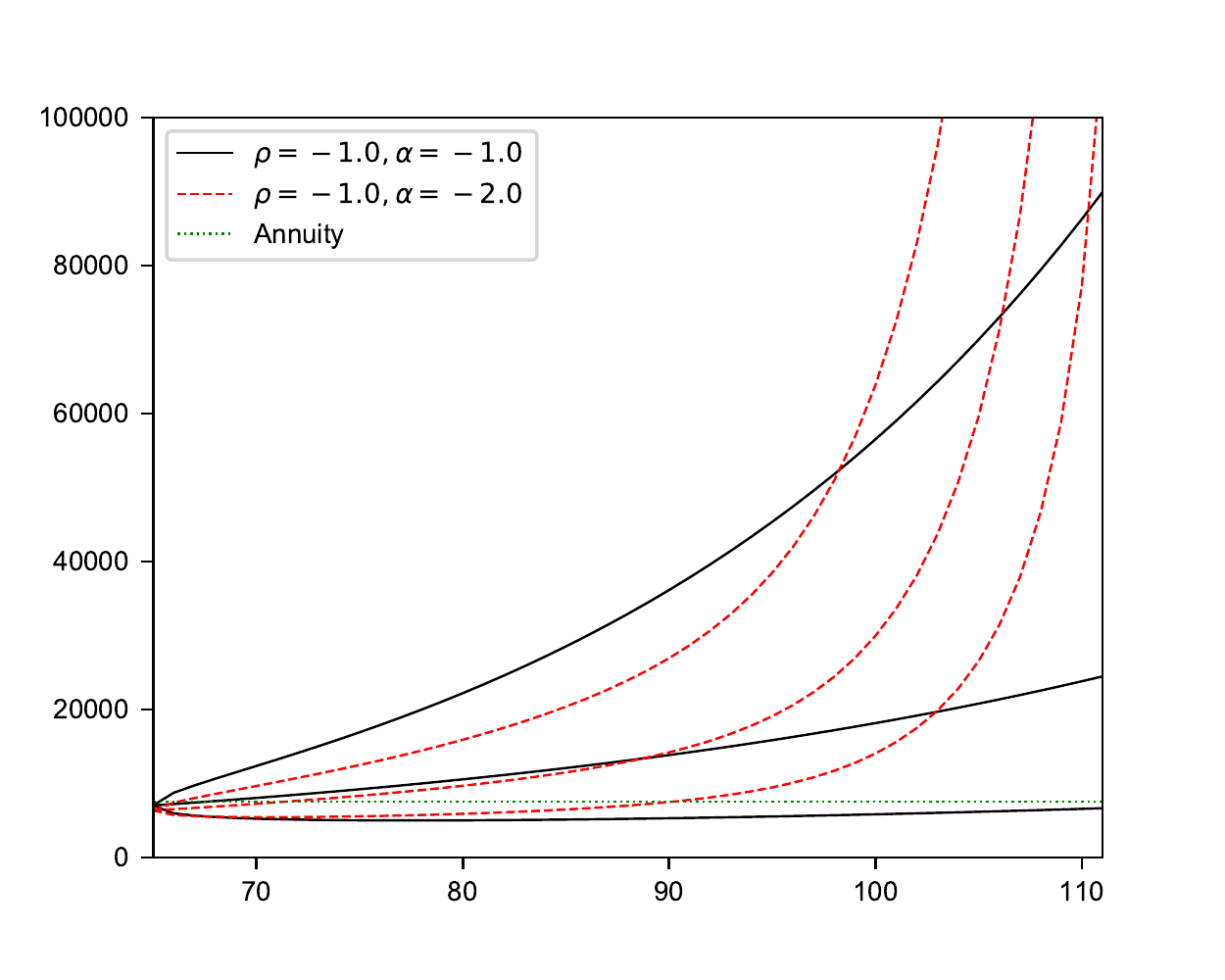}
	\caption{Fan diagram of optimal consumption for an infinite collective ($n=\infty$) with homogeneous Epstein--Zin preferences.
		We illustrate both the case of intertemporally additive von Neumannn--Morgernstern preferences ($\alpha=\rho$)
		and of satiation-risk-aversion ($\alpha<\rho$). The percentiles in the fans
		are $(5\%, 50\%, 95\%)$. }
	\label{fig:epstein-zin-with-riskaversion}
\end{figure}

To compare the results we have seen with 
those obtained under homogeneous Epstein--Zin preferences, 
we plot in Figure \ref{fig:epstein-zin-with-riskaversion} the optimal consumption calculated
using homogeneous Epstein--Zin preferences, but with all market parameters as before. The state pension
and the adequacy level $AL_t$ are no longer used in the calculation. Instead, the adequacy level
will be determined by the coefficients $\alpha$ and $\rho$ themselves and will always
take a value of $\infty$ or $0$. The analytic results of \cite{ab-ez} were
used to plot this diagram.

The marked qualitative difference between the plot for exponential Kihlstrom--Mirman preferences
and for homogeneous Epstein--Zin preferences can be explained by the tendency to track the
adequacy level as satiation-risk-aversion increase. The median consumption gradually decreases in very old age for the collective fund with exponential Kihlstrom--Mirman preferences in Figure \ref{fig:obrAssumptionsFan} because the adequacy
level in this case is zero. The median consumption increases in very old age for 
the collective fund with homogeneous Epstein--Zin preferences in Figure \ref{fig:epstein-zin-with-riskaversion}
because the adequacy level in this case is $\infty$.

We believe this emphasizes that considering the adequacy level and state pension is important
in determining the optimal investment strategy.

One can incorporate the state pension and adequacy level into Epstein--Zin preferences if
one is willing to sacrifice homogeneity and to choose a value for the parameter $\beta$
required for inhomogeneous Epstein--Zin preferences. We believe that the numerical method
of \cite{ab-exponential} could be adapted to solve the problem in this case.

Our results indicate that homogeneous Epstein--Zin preferences with non-zero satiation risk aversion are not a good model for practical investment decisions.
Nevertheless we do not intend to dismiss this preference model entirely as we
believe the analytic results of \cite{ab-ez} yield considerable insights. They yield a stylised view of the optimal strategy where the effects of adequacy are exaggerated due to the choice of an extreme value for the adequacy level.

\section{Managing a Heterogeneous Fund}
\label{sec:heterogeneousNumerics}

Our results so far have considered only the case of homogeneous funds. The
problem of managing a heterogeneous fund is studied in \cite{ab-heterogeneous}. 
That paper uses an axiomatic approach to bound the utility that an individual can achieve from a collectivised investment. The argument proceeds by observing
that in a complete market there are no mutually beneficial contracts between two investors without mortality. One may then deduce
that there are no mutually beneficial contracts between two infinitely large collectives of identical investors as these
collectives are not subject to mortality risk.
On the other hand, collectivising one's pension with similar individuals is always beneficial. These observations can be combined to show that the value of the gain function that any individual can hope to achieve when investing in a heterogeneous fund is bounded above by the value of the gain function they can achieve when investing in an infinite collective.

Given this intuition, we now propose a practical algorithm that should achieve this bound as the number of investors in the heterogeneous fund tends to infinity. We may then test numerically how the algorithm performs for
small funds. 

\begin{algorithm}[Heterogeneous fund algorithm]
	\label{algo:heterogeneous}
	Choose a positive integer $n_{\max}$. Let $n_t$ denote the number of survivors at time $t$. Define $n^\prime_t$ by
	\[
	n_t^\prime = \begin{cases}
	n_t & n_t \leq n_{\max} \\
	\infty & \text{otherwise}.
	\end{cases}
	\]
	Compute the consumption and investment strategy as follows:
	\begin{enumerate}
		\item  Keep accounts of the current funds associated with
		each individual.
		\item At time $t$, for each surviving individual $i$ in the fund, invest and consume
		according to the optimal strategy for a homogeneous fund of $n_t^\prime$
		investors identical to that individual with budget given by their current
		funds. Even if, after consumption, the individual dies at time $t$,
		one should pursue the
		same investment strategy as one would have done if they had survived.
		
		Compute the resulting wealth $\mathring{X}^i_{t+\delta t}$ of each individual
		$i$ who was alive at time $t$.
		\item For an investor $i$ who survives to time $t+\delta t$ we define their
		``contribution'' to the collective at time $t+\delta t$, $\Gamma^i_{t+\delta t}$
		by
		\[
		\Gamma^i_{t+\delta t}
		= (1-s^i_t) \mathring{X}^i_{t+\delta t}
		\]
		where $s^i_t$ is the survival property of individual $i$ from time $t$ to $t+\delta t$.
		This can be interpreted as a fair price for the derivative contract with a payoff
		equal to the wealth received by the collective if the individual dies and zero
		otherwise, so long as the pricing measure for the individual's mortality is taken
		to equal the physical measure $\P$.
		\item When an individual dies, divide their funds among the survivors
		in proportion to each survivor's contribution $\Gamma^i_{t+\delta t}$.	      
	\end{enumerate}
\end{algorithm}

The purpose of the cut-off $n_{\max}$ in this algorithm is simply that it
is computationally expensive to compute the optimal strategy for a collective of $n$ investors if $n$
is large.

The logic behind this algorithm is that we assume we can divide
our population into large groups of similar individuals.
Let one such group of individuals be close to one particular individual
which we label $\zeta$.
Since the individuals are similar we assume that the optimal
strategies for each member of the group will be similar. If the number
of individuals in the group is large, the optimal strategy for $n_t$ individuals
of a given type will be similar to the optimal strategy for $\infty$ individuals
of a given type. The number of survivors
will also be close to the expected value. Thus (2) and (3)
together will ensure that the utility achieved by individuals of type close
to $\zeta$ will be close to the utility that can be obtained by an infinite
collective of individuals all of type exactly $\zeta$.

This argument is essentially a compactness and
continuity argument. One could therefore write out a formal topological
proof of its convergence as the number of individuals tends to infinity,
under suitable assumptions. However, we do not believe doing so
would be particularly illuminating.

It is natural to ask just how large a fund is required in order that this algorithm achieves a value of the gain function that is close to this upper bound for all investors. To answer this question we generated a random fund
of $n=100$ individuals. Each individual had inter-temporally additive von Neumann--Morgernstern 
preferences given by a power utility with the power uniformly generated
in the range $[-1.5,-0.5]$. The initial wealth of each individual
was taken to be uniformly distributed in the range $[0.5,1.5]$.
The retirement age of each individual was taken to be a uniformly
distributed random integer in the range $60$--$69$. The sex
of the individual was chosen as male with probability $50\%$. All of these
random choices were made independently.
The mortality distribution for each individual
was then computed using the CMI model with longterm rate $1.5\%$
and retirement year of $2019$. The market parameters were chosen as in Section \ref{sec:epsteinZin}.

Using this population, we ran $10^6$ simulations of 
the algorithm above with $n_{\max}$ set to $50$.. This allowed us to compute a sample
average utility, $u^i_S$, for each individual $i$.

We define $u^i_n$ to be the expected utility that would be achieved by
individual $i$ if they were to invest in a homogeneous collective of $n$ 
individuals. We define the {\em optimality ratio} for individual $i$ by
\[
\text{OR}_i := \frac{u^i_S-u^i_1}{u^i_\infty-u^i_1}.
\] 
If this ratio is close to $1$, then the utility experienced by individual $i$
is close to the optimum value they can expect from any acceptable collectivised investment.

In Figure \ref{fig:heterogeneous-histogram} we plot a histogram of the optimality ratio, $\text{OR}_i$, for each of our $100$
fund members. In our example,
the optimality ratio is almost always above $98\%$. This demonstrates
that, even with as few as $100$ investors, our investment strategy is close to
optimal.

We notice that some individuals are lucky, and for them our strategy results in them receiving a slightly higher
gain than the maximum predicted in \cite{ab-heterogeneous}. This does not contradict the results of that paper
as our algorithm only approximately satisfies the formal axioms given there. In particular it is possible that with our axiom some investors might in principle be better off forming a smaller collective rather than joining with a single collective fund. Our numerical results shows that the potential gain of doing this will be marginal, and this is the sense in which the axioms of \cite{ab-heterogeneous} are approximately satisfied.

\begin{figure}
	\centering
	\includegraphics{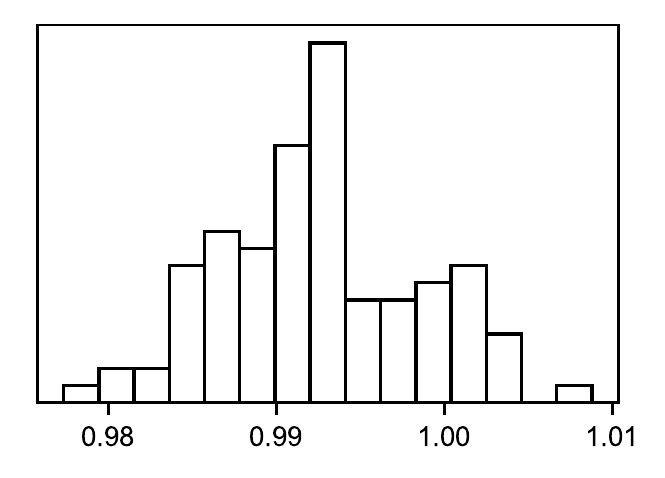}
	\caption{Histogram of the optimality ratio, ${OR}_i$, 
    obtained for a
	randomly generated fund of $100$ investors.}
	\label{fig:heterogeneous-histogram}
\end{figure}

\section{Conclusion}
\label{sec:conclusion}

We find that, given the choice, one should not invest one's pension from a defined contribution
pension fund
in an annuity. Instead one should invest in a collectivised investment fund. To match the performance
of a collectivised investment fund using an annuity, one requires $20\%$ more initial capital.
It follows that there should be a considerable market for collectivised investment funds. Our results 
indicate that such a fund would not need to be exceptionally large in order to be viable.

It also follows from our work that an optimally managed Collective Defined Contribution pension fund will outperform a Defined Benefit pension fund. Thus collective schemes have the potential to provide a new form of pension, to the mutual benefit of both employee and employer.


\bibliographystyle{plain}
\bibliography{collectivization}

\end{document}